\documentclass[11pt]{article}
\usepackage{xspace}
\usepackage{mathptmx}
\usepackage{epsfig}
\usepackage{epstopdf}
\usepackage{graphicx}
\usepackage{amsfonts}
\usepackage{amsmath}
\usepackage{amsthm}
\usepackage{proof}
\usepackage{latexsym}
\usepackage{amssymb}
\usepackage{color}
\usepackage{comment}
\usepackage[margin=1in]{geometry}
\usepackage{extras,tikz}
\usepackage{psfrag}
\usepackage{comment}
\usepackage[noend]{algorithmic}
\usepackage{algorithm}
\usepackage{tweaklist}

\newtheorem{theorem}{Theorem}[section]
\newtheorem{lemma}[theorem]{Lemma}

\newtheorem{definition}[theorem]{Definition}

\newtheorem{remark}[theorem]{Remark}
\newtheorem{conjecture}[theorem]{Conjecture}

\DeclareMathAlphabet{\mathcal}{OMS}{cmsy}{m}{n}

\newcommand {\NP} {\mbox{NP}}
\newcommand {\SUBEXP} {\mbox{SUBEXP}}
\newcommand {\FPT} {\mbox{FPT}}

\newcommand{\OPT}{\textup{OPT}\xspace}

\newcommand{\dst}{\mbox{Directed Steiner Tree}\xspace}
\newcommand{\dsf}{\mbox{Directed Steiner Forest}\xspace}
\newcommand{\dsn}{\mbox{Directed Steiner Network}\xspace}

\newcommand{\mec}{\mbox{Minimum Size Edge Cover}\xspace}
\newcommand{\mcc}{\mbox{Multicolored Clique}\xspace}
\newcommand{\scss}{\mbox{Strongly Connected Steiner Subgraph}\xspace}

\newcommand*\samethanks[1][\value{footnote}]{\footnotemark[#1]}

\begin{document}

\title{Fixed-Parameter and Approximation Algorithms: A New Look}
\author{Rajesh Chitnis  \thanks{Department of Computer Science, University of Maryland at College Park, USA. Supported in part by NSF CAREER award 1053605, NSF grant CCF-1161626,
      ONR YIP award N000141110662, DARPA/AFOSR grant FA9550-12-1-0423,
      and a University of Maryland Research and Scholarship Award (RASA). The first author is also supported by
      a Simons Award for Graduate Students in Theoretical Computer Science. The second author is also with AT\&T Labs. Email: \{\texttt{rchitnis,
hajiagha\}@cs.umd.edu}}
  \and
  MohammadTaghi Hajiaghayi\samethanks[1]
 \and Guy Kortsarz\thanks{Rutgers University, Camden, NJ. Supported in part by NSF grant number 434923. Email: \texttt{guyk@camden.rutgers.edu}}
}



\date{\today}

\maketitle
\begin{abstract}
A Fixed-Parameter Tractable (\FPT) $\rho$-approximation algorithm for a minimization (resp. maximization) parameterized problem
$P$ is an FPT algorithm that, given an instance $(x, k)\in P$ computes a solution of cost at most $k \cdot \rho(k)$ (resp.
$k/\rho(k)$) if a solution of cost at most (resp. at least) $k$ exists; otherwise the output can be arbitrary. For well-known
intractable problems such as the W[1]-hard \mbox{Clique} and W[2]-hard \mbox{Set Cover} problems, the natural question is
whether we can get any \FPT-approximation.
It is widely believed that both \mbox{Clique} and \mbox{Set-Cover} admit no FPT $\rho$-approximation algorithm, for any
increasing function $\rho$. However, to the best of our knowledge, there has been no progress towards proving this conjecture.
Assuming standard conjectures such as the Exponential Time Hypothesis (ETH)~\cite{eth-paturi} and the Projection Games
Conjecture (PGC)~\cite{r3}, we make the first progress towards proving this conjecture by showing that
\begin{itemize}
\item Under the ETH and PGC, there exist constants $F_1, F_2 >0$ such that the \mbox{Set Cover} problem does not admit a
    FPT approximation algorithm with ratio $k^{F_1}$ in $2^{k^{F_2}}\cdot \text{poly}(N,M)$ time, where $N$ is the size of
    the universe and $M$ is the number of sets.
\item Unless $\NP\subseteq \SUBEXP$, for every $1> \delta > 0$ there exists a constant $F(\delta)>0$ such that
    \mbox{Clique} has no FPT cost approximation with ratio $k^{1-\delta}$ in $2^{k^{F}}\cdot \text{poly}(n)$ time, where
    $n$ is the number of vertices in the graph.
\end{itemize}
In the second part of the paper we consider various W[1]-hard problems such as {\dst}, {\dsf}, Directed Steiner
Network and {\mec}.
For all these problem we give polynomial time $f(\text{OPT})$-approximation algorithms for some small function $f$ (the
largest approximation ratio we give is $\text{OPT}^2$).
Our results indicate a potential separation between the classes W[1] and W[2]; since no W[2]-hard problem is known to have a
polynomial time $f(\text{OPT})$-approximation for any function $f$. Finally, we answer a question by Marx~\cite{daniel-survey}
by showing the well-studied Strongly Connected Steiner Subgraph problem (which is W[1]-hard and does not have any polynomial
time constant factor approximation) has a constant factor FPT-approximation.

\end{abstract}

\section{Introduction}
\label{sec:intro}

\emph{Parameterized Complexity} is a two-dimensional generalization of ``P vs. NP'' where in addition to the overall input
size $n$, one studies the effects on the computational complexity of a secondary measurement that captures additional relevant
information. This additional information can be, for example, a structural restriction on the input distribution considered,
such as a bound on the treewidth of an input graph or the size of a solution. For general background on the theory
see~\cite{downey-fellows}. For decision problems with input size $n$, and a parameter $k$, the two dimensional analogue (or
generalization) of P, is solvability within a time bound of $O(f(k)n^{O(1)})$, where $f$ is a function of $k$ alone.
Problems having such an algorithm are said to be \emph{fixed parameter tractable} (FPT). The $W$-hierarchy is a collection of
computational complexity classes: we omit the technical definitions here. The following relation is known amongst the classes
in the $W$-hierarchy: $\text{FPT}=W[0]\subseteq W[1]\subseteq W[2]\subseteq \ldots$. It is widely believed that $\FPT\neq \text{W[1]}$, and
hence if a problem is hard for the class $W[i]$ (for any $i\geq 1$) then it is considered to be fixed-parameter intractable.
We say that a problem is W-hard if it is hard for the class W[i] for some $i\geq 1$. When the parameter is the size of the
solution then the most famous examples of W[1]-hard and W[2]-hard problems are \mbox{Clique} and \mbox{Set Cover}
respectively. We define these two problems below:
\begin{center}
\noindent\framebox{\begin{minipage}{6.0in}
\mbox{\textbf{Clique}}\\
\emph{Input }: An undirected graph $G=(V,E)$, and an integer $k$\\
\emph{Problem}: Does $G$ have a clique of size at least $k$?\\
\emph{Parameter}: $k$
\end{minipage}}
\end{center}

\begin{center}
\noindent\framebox{\begin{minipage}{6.0in}
\mbox{\textbf{Set Cover}}\\
\emph{Input}: Universe $U=\{u_1,u_2,\ldots,u_n\}$  and a collection $\mathcal{S}=\{S_1,S_2,\ldots,S_m\}$ of subsets of $U$
such that
$\bigcup_{j=1}^{m} S_j = U $.\\
\emph{Problem}: Is there a subcollection $\mathcal{S}'\subseteq \mathcal{S}$ such that $\mathcal{S}'\leq  k$ and
        $\bigcup_{S_i\in \mathcal{S}'} S_i= U$?\\
\emph{Parameter}: $k$
\end{minipage}}
\end{center}

The next natural question is whether these fixed-parameter intractable problems at least admit parameterized approximation
algorithms.

\subsection{Parameterized Approximation Algorithms}
\label{sec:intro-parameterized-approximation}

We follow the notation from Marx~\cite{marx-ccc2010}. Any \NP-optimization problem can be described as $O=(I, \text{sol}, \text{cost},
\text{goal})$, where $I$ is the set of instances, $\text{sol}(x)$ is the set of feasible solutions for instance $x$, the
positive integer $\text{cost}(x; y)$ is the cost of solution $y$ for instance $x$, and goal is either min or max. We assume
that $\text{cost}(x, y)$ can be computed in polynomial time, $y\in \text{sol}(x)$ can be decided in polynomial time, and $|y|$
= $|x|^{O(1)}$ holds for every such $y$.

\begin{definition}
\label{defn-1} Let $\rho : \mathbb{N} \rightarrow \mathbb{R}_{\geq 1}$ be a computable function such that $\rho(k)\geq 1$ for
every $k\geq 1$; if \emph{goal}=\emph{min} then $k\cdot \rho(k)$ is nondecreasing and if the \emph{goal}=\emph{max} then $k/\rho(k)$ is unbounded and
nondecreasing. An \textbf{FPT approximation algorithm} with approximation ratio $\rho$ for $O$ is an algorithm $\mathbb{A}$
that, given an input $(x, k)\in \Sigma^{*}\times \mathbb{N}$ satisfying $\emph{sol}(x)\neq \emptyset$ and
\begin{equation}\label{eqn:star}
   \begin{cases}
   opt(x)\leq k & \text{if}\ \emph{goal}=\emph{min} \\
   opt(x)\geq k & \text{if}\ \emph{goal}=\emph{max}
    \end{cases}
\tag{*}
\end{equation}
computes $y\in \emph{sol}(x)$ such that
\begin{equation}\label{eqn:star-star}
   \begin{cases}
   \emph{cost}(x,y)\leq k\cdot \rho(k) & \text{if}\ \emph{goal}=\emph{min} \\
   \emph{cost}(x,y)\geq k/\rho(k) & \text{if}\ \emph{goal}=\emph{max}
    \end{cases}
\tag{**}
\end{equation}
We require that on input $(x,k)$ the algorithm $\mathbb{A}$ runs in $f(k)\cdot |x|^{O(1)}$ time, for some computable function
$f$.
\end{definition}

Note that if the input does not satisfy (*), then the output can be arbitrary.

\begin{remark}
\label{normalized} Given an output $y\in \emph{sol}(x)$ we can check in FPT time if it satisfies \emph{(**)}. Hence we can assume
that an FPT approximation algorithm always\footnote{even if the input does not satisfy (*)} either outputs a $y\in \emph{sol}(x)$
that satisfies \emph{(**)} or outputs a default value \texttt{reject}. We call such an FPT approximation algorithm that has
this property as \textbf{normalised}.
\end{remark}

Classic polynomial-time approximation algorithms determine the performance ratio by comparing the output with the optimum. In
FPT approximation algorithms there is a subtle difference: we compare the output to the parameter to determine the
approximation ratio. Fellows~\cite{dagstuhl-fellows} asked about finding an FPT approximation algorithm for W[2]-hard
Dominating Set (which is a special case of Set Cover), or ruling out such a possibility. The following conjecture due to Marx (personal communication) is widely believed in the FPT community:

\begin{conjecture}
\label{conj:clique-and-set-cover-no-fpt-approx} Both \mbox{Set Cover} and \mbox{Clique} do not admit an FPT algorithm with
approximation ratio $\rho$, for any function $\rho$.
\end{conjecture}

However to the best of our knowledge there has been no progress towards proving this conjecture, even under assumptions from
complexity theory. In this paper we take a first step towards proving
Conjecture~\ref{conj:clique-and-set-cover-no-fpt-approx}, under well-known and reasonable\footnote{It is very important to only work under well-believed assumptions, since otherwise we will be able to prove pretty much what we want, but it is of no value.} assumptions from complexity theory
like the Exponential Time Hypothesis (ETH) of Impagliazzo et al.~\cite{eth-paturi} and the Projection Games Conjecture (PGC)
of Moshkovitz~\cite{r3}.

For both minimization and
maximization problems, the most interesting and practical case is the input $(x,k)$ when $k=OPT(x)$. This motivates the
definition of the following variant of FPT approximation algorithms:

\begin{definition}
\label{defn-opt} Let $\rho : \mathbb{N} \rightarrow \mathbb{R}_{\geq 1}$ be a computable function such that $\rho(k)\geq 1$
for every $k\geq 1$; if \emph{goal}=\emph{min} then $k\cdot \rho(k)$ is nondecreasing and if \emph{goal}=\emph{max} then $k/\rho(k)$ is unbounded and
nondecreasing. An \textbf{FPT optimum approximation algorithm} for $O$ with approximation ratio $\rho$ is an algorithm
$\mathbb{A}'$ that, given an input $x\in \Sigma^{*}$ satisfying $\emph{sol}(x)\neq \emptyset$ outputs a $y\in \emph{sol}(x)$ such that
\begin{equation}\label{eqn:star'}
   \begin{cases}
   \emph{cost}(x,y)\leq OPT(x)\cdot \rho(OPT(x)) & \text{if}\ \emph{goal}=\emph{min} \\
   \emph{cost}(x,y)\geq OPT(x)/\rho(OPT(x)) & \text{if}\ \emph{goal}=\emph{max}
    \end{cases}
\end{equation}
We require that on input $x$ the algorithm $\mathbb{A}$ runs in $f(OPT(x))\cdot |x|^{O(1)}$ time, for some computable function
$f$.
\end{definition}

In Section~\ref{sec:thm-mini-implication}, we show the following theorem:
\begin{theorem}
\label{thm:mini-implication} Let $O$ be a minimization problem in \NP and $\mathbb{A}$ be an FPT approximation algorithm for $O$ with ratio $\rho$. On input $(x,k)$ let the running time of $\mathbb{A}$ be $f(k)\cdot |x|^{O(1)}$ for some non-decreasing computable function $f$. Then $O$ also has an FPT optimum approximation algorithm $\mathbb{A}'$ with approximation ratio $\rho$, and whose running time on input $x$ is also $f(OPT(x))\cdot |x|^{O(1)}$
\end{theorem}

Hence for minimization problems, it is enough to prove hardness results only for the notion of FPT optimum approximation algorithms (see Definition~\ref{defn-opt}).
We do not know any relation between the two definitions for maximization problems, and hence we prove hardness results for both Definition~\ref{defn-1} and Definition~\ref{defn-opt}.

\section{Our Results} \label{sec:our-results}

We make the first progress towards proving Conjecture~\ref{conj:clique-and-set-cover-no-fpt-approx}, under
standard assumptions from complexity theory. In particular for \mbox{Set Cover} we assume the Exponential Time Hypothesis
(ETH)~\cite{eth-paturi} and the Projection Games Conjecture (PGC)~\cite{r3}\footnote{The PGC is stated in
Section~\ref{sec:pgc}, and the ETH is stated in Section~\ref{sec:eth}}. The PGC gives a reduction from SAT to Projection
Games. Composing this with the standard reduction from Projection Games to \mbox{Set Cover} gives a reduction from SAT to
\mbox{Set Cover}. Since the ETH gives a lower bound on the running time of SAT, we are able to show the following
inapproximability result in Section~\ref{sec:lower-bound-set-cover}:

\begin{theorem}
\label{setcover} Under the ETH and PGC,
\begin{enumerate}
\item There exist constants $F_1, F_2 >0$ such that the \mbox{Set Cover} problem does not admit an FPT \textbf{optimum}
    approximation algorithm with ratio $\rho(OPT)=OPT^{F_1}$ in $2^{OPT^{F_2}}\cdot \text{poly}(N,M)$ time, where $N$ is
    the size of the universe and $M$ is the number of sets.
\item There exist constants $F_3, F_4 >0$ such that the \mbox{Set Cover} problem does not admit an FPT approximation
    algorithm with ratio $\rho(k)=k^{F_3}$ in $2^{k^{F_4}}\cdot \text{poly}(N,M)$ time, where $N$ is the size of the
    universe and $M$ is the number of sets.
\end{enumerate}
\end{theorem}

In Section~\ref{sec:lower-bound-clique}, we consider the \mbox{Clique} problem. We use the result of
Zuckerman~\cite{Zu} which states that it is \NP-hard to get an $O(n^{1-\epsilon})$-approximation for \mbox{Clique}. Given any
problem $X\in \NP$, by using the Zuckerman reduction from $X$ to \mbox{Clique} allows us to show the following result.

\begin{theorem}
\label{clique} Unless $\NP\subseteq \SUBEXP$, for every $1> \delta > 0$
\begin{enumerate}
\item There exists a constant $F(\delta)>0$ such that \mbox{Clique} has no \FPT\ \textbf{optimum} approximation with ratio
    $\rho(OPT)=OPT^{1-\delta}$ in $2^{OPT^{F}}\cdot \text{poly}(n)$ time, where $n$ is the number of vertices in the
    graph.
\item There exists a constant $F'(\delta)>0$ such that \mbox{Clique} has no \FPT\ approximation with ratio
    $\rho(k)=k^{1-\delta}$ in $2^{k^{F'}}\cdot \text{poly}(n)$ time, where $n$ is the number of vertices in the graph.
\end{enumerate}

\end{theorem}

\subsection{Polytime $f(OPT)$-approximation for W-hard problems}
\label{sec:intro-f-opt-poly-time}

We also deal with the following question: given that a problem is W-hard, can we maybe get a good polynomial-time
approximation for the problem? Any problem in $NP$ can be solved in $2^{n^{O(1)}}$ time by simply enumerating all candidates
for the witness. If the parameter $k$ is at least $\log n$, then we immediately have $2^{k}\geq n$ and the problem can be
solved in $2^{n^{O(1)}}\leq 2^{{2^k}^{O(1)}}$ time which is FPT time in $k$. So for large values of the parameter the brute
force algorithm itself becomes an FPT algorithm. Hence the intrinsic hardness to obtain FPT algorithms for intractable
problems is when the parameter $k$ is small (say at most $\log n$).
In this case, we show how to replace the impossible FPT solution by a good approximation, namely $f(OPT)$ approximation for some small
function $f$.
We systematically design polynomial-time $f(OPT)$
approximation algorithms for a number of W[1]-hard minimization problems such as {\mec}, Strongly Connected Steiner Subgraph, {\dsf} and
Directed Steiner Network. Each of the aformentioned problems is known to have strong inapproximability
(in terms of input size). Since we can assume $OPT$ is small, this implies $f(OPT)$ is small as well. Therefore for these
W[1]-hard problems, if the parameter is large then we can get an FPT algorithm, otherwise if the parameter is small (then OPT is small as well, otherwise we can reject for these minimization problems)and we obtain a reasonable approximation in polynomial time. These results point towards a separation between the classes W[1] and
W[2] since we do not know any W[2]-hard problem which has a polynomial-time $f(OPT)$-approximation, for any function $f$. In
fact, Marx (personal communication) conjectured that the W[2]-hard \mbox{Set Cover} problem does not have a polynomial-time
$f(OPT)$-approximation for any function $f$.

Finally in Section~\ref{sec:scss} we show that the well-studied W[1]-hard Strongly Connected Steiner Subgraph problem has an
FPT $2$-approximation algorithm. This answers a question by Marx~\cite{daniel-survey} regarding finding a problem which is
fixed-parameter intractable, does not have a constant factor approximation in polynomial time but admits a constant factor
FPT approximation. To the best of our knowledge no such W[2]-hard problem (parameterized by solution size) is known,
and this indicates another potential difference between W[1] and W[2].

\subsection{Proof of Theorem~\ref{thm:mini-implication}}
\label{sec:thm-mini-implication}

Let $x\in \Sigma^*$ be the input for $\mathbb{A}'$. The algorithm $\mathbb{A}'$ runs the algorithm $\mathbb{A}$ on the instances $(x,1), (x,2), \ldots$ until the first $k$ such that the output of $\mathbb{A}$ on $(x,k)$ is a solution of cost at most $k\cdot \rho(k)$. Then $\mathbb{A}'$ outputs $\mathbb{A}(x,k)$. By Definition~\ref{defn-1}, we know that $k\leq OPT(x)$. Hence $k\cdot \rho(k)\leq OPT(x)\cdot \rho(OPT(x))$. It remains to analyze the running time of $\mathbb{A}'$.

Since $k\leq OPT(x)$, the running time of $\mathbb{A}'$ is upper bounded by $\sum_{i=1}^{k} f(i)\cdot |x|^{O(1)}\leq \sum_{i=1}^{OPT(x)} f(i)\cdot |x|^{O(1)}= \Big(\sum_{i=1}^{OPT(x)} f(i)\Big)\cdot |x|^{O(1)} \leq OPT(x)\cdot f(OPT(x))\cdot |x|^{O(1)} = f(OPT(x))\cdot |x|^{O(1)}$, since $f$ is non-decreasing and $OPT(x)\leq |x|$.
\qed

\section{Conjectures from Computational Complexity}

In this section, we describe two conjectures from computational complexity that we work with in this paper.

\subsection{Exponential Time Hypothesis}
\label{sec:eth}

Impagliazzo, Paturi and Zane~\cite{eth-paturi} formulated the following conjecture which is known as the Exponential Time
Hypothesis (ETH).

\begin{center}
\noindent\framebox{\begin{minipage}{6.00in}
\textbf{Exponential Time Hypothesis} (ETH)\\
$3$-$SAT$ cannot be solved in $2^{o(n)}$ time where $n$ is the number of variables.
\end{minipage}}
\end{center}

Using the Sparsification Lemma of Calabro, Impagliazzo and Paturi~\cite{sparsification-paturi}, the following lemma was shown
in~\cite{eth-paturi}.

\begin{lemma}
\label{lem:sparsification-eth} Assuming ETH, $3$-$SAT$ cannot be solved in $2^{o(m)}$ time where $m$ is the number of clauses.
\end{lemma}

In the reductions from $3$-$SAT$ to Clique, Vertex Cover and Independent Set, the number of vertices formed in the graphs is
equal to the number of clauses in the $3$-$SAT$ instance and hence Lemma~\ref{lem:sparsification-eth} gives  evidence against
subexponential algorithms for the above three problems. This is enough to give some belief in the ETH. We note that ETH and
its variants have been used to prove lower bounds in both FPT~\cite{fpt-eth} and exact exponential
algorithms~\cite{exact-eth}. We refer to~\cite{survey-eth} for a nice survey on lower bounds using ETH. In this paper, we use
ETH to give inapproximability results for Set Cover.\footnote{It is not clear what Moshkovitz~\cite{r3} refers to as the size of a SAT instance. If it the number of variables, then we use the ETH as is. Otherwise if it refers to the number of clauses, then we are still fine by the Sparsification Lemma~\cite{sparsification-paturi}}

\subsection{The Projection Games Conjecture}
\label{sec:pgc}

First we define a \emph{projection game}. Note that with a loss of factor two we can assume that the alphabet is the same for
both sides. The input to a projection game consists of:
\begin{itemize}
\item A bipartite graph $G = (V_1,V_2,E)$
\item  A finite alphabets $\Sigma$
\item Constraints (also called projections) given by $\pi_{e}:\Sigma \rightarrow \Sigma$ for every $e\in E$.
\end{itemize}
The goal is to find an assignment $\phi: V_1\cup V_2\rightarrow \Sigma$ that \emph{satisfies} as many of the edges as
possible. We say that an edge $e = \{a,b\}\in E$ is satisfied, if the projection constraint holds, i.e., $\pi_{e}(\phi(a)) =
\phi(b)$. We denote the size of a projection game by $n = |V_1| + |V_2| + |E|$.

\begin{conjecture}
\label{conj:pgc} (\emph{Projection Games Conjecture}~\cite{r3}) There exists $c > 0$ such that for every $\epsilon>0$, there
is a polynomial reduction \texttt{RED-1} from SAT\footnote{SAT is the standard Boolean satisfiability problem} to Projection Games which maps an instance $I$ of SAT to an instance $I_1$
of Projection Games such that:
\begin{enumerate}
\item If a YES instance $I$ of SAT satisfies $|I|^c\geq \frac{1}{\epsilon}$, then all edges of $I_1$ can be satisfied.
\item If a NO instance $I$ of SAT satisfies $|I|^c\geq \frac{1}{\epsilon}$, then at most $\epsilon$-fraction of the edges
    of $I_1$ can be satisfied.
\item The size of $I_1$ is almost-linear in the size of $I$, and is given by $|I_1|=n= |I|^{1+o(1)}\cdot
    \emph{poly}(\frac{1}{\epsilon})$.
\item The alphabet $\Sigma$ for $I_1$ has size $\emph{poly}(\frac{1}{\epsilon})$.
\end{enumerate}
\end{conjecture}

A weaker version of the conjecture is
known, but the difference is that the alphabet in~\cite{dana} has size $\text{exp}(\frac{1}{\epsilon})$. As pointed out
in~\cite{r3}, the Projection Games Conjecture is an instantiation of the Sliding Scale Conjecture of Bellare et al.~\cite{c1}
from 1993. Thus, in fact this conjecture is actually 20 years old. But we have reached a state of knowledge now that it seems
likely that the Projection Games Conjecture will be proved not long from now (see Section 1.2 of~\cite{r3}). Thus it seems that posing this
conjecture is quite reasonable. In contrast to this is the Unique Games Conjecture \cite{ugc}. On the positive side, it seems
that the Unique Games Conjecture is much more influential than the Projection Games Conjecture. But it seems unlikely (to us)
that the Unique Games Conjecture will be resolved in the near future.

\section{An FPT Inapproximability Result for \mbox{Set Cover}}
\label{sec:lower-bound-set-cover}

The goal of this section is to prove Theorem~\ref{setcover}.

\subsection{Reduction from Projection Games to Set Cover}
The following reduction from Projection Games to Set Cover is known, see~\cite{dorit,DBLP:journals/jacm/LundY94}. We sketch a
proof below for completeness.

\begin{theorem}
\label{redn:proj-set-cover} There is a reduction \texttt{RED-2} from Projection Games to \mbox{Set Cover} which maps an
instance $I_1=(G=(V_1,V_2,E), \Sigma, \pi)$ of Projection Games to an instance $I_2$ of \mbox{Set Cover} such that:
\begin{enumerate}
\item If all edges of $I_1$ can be satisfied, then $I_2$ has a set cover of size $|V_1|+|V_2|$.
\item If at most $\epsilon$-fraction of edges of $I_1$ can be satisfied, then the size of a minimum set cover for $I_2$ is
    at least $\frac{|V_1|+|V_2|}{\sqrt{32\epsilon}}$
\item The instance $I_2$ has $|\Sigma|\times (|V_1|+|V_2|)$ sets and the size of the universe is
    $2^{O(\frac{1}{\sqrt{\epsilon}})} \times |\Sigma|^{2}\times |E|$
\item The time taken for the reduction is upper bounded by $2^{O(\frac{1}{\sqrt{\epsilon}})} \times \emph{poly}(|\Sigma|)\times \emph{poly}(|E|+|V_1|+|V_2|)$
\end{enumerate}
\end{theorem}

\subsubsection{Proof of Theorem \ref{redn:proj-set-cover}}
\label{1}

\begin{definition}
\label{defn:ml-set-systems} An $(m, \ell)$-set system consists of a universe $B$ and collection of subsets
$\{C_1,\ldots,C_m\}$ such that if the union of any sub-collection of $\ell$-sets from the collection
$\{C_1,\ldots,C_m,\overline{C_1},\ldots,\overline{C_m}\}$ is $B$, then the collection must contain both $C_i$ and
$\overline{C_i}$ for some $i\in [m]$.
\end{definition}

It is known that an $(m, \ell)$-set system with a universe size $|B| = O(2^{2\ell}m^2)$ exists, and can be constructed in
$2^{O(\ell)}\cdot m^{O(1)}$ time~\cite{dorit}. Consider the following reduction:

\begin{center}
\noindent\framebox{\begin{minipage}{6.00in}
\mbox{\textbf{Projection Games Instance}}:\\
$(G=(V_1,V_2,E), \Sigma, \pi)$ such that $|\Sigma|=m$. \\
~\\
\mbox{\textbf{Set Cover Instance}}:\\
Let $B$ be a $(m, \ell)$ set system. The universe for the set cover instance consists of $E\times B$. Define the following
subsets of $E\times B$
\begin{itemize}
\item For all vertices $v\in V_2$ and $x\in \Sigma$, define the subset $S_{v,x}= \bigcup_{e\ni v} \{e\} \times C_x$
\item For all vertices $u\in V_1$ and $y\in \Sigma$, define the subset $S_{u,y}= \bigcup_{e\ni u} \{e\} \times
    \overline{C_{\pi_{e}(y)}}$
\end{itemize}
The Set Cover instance produced is $(E\times B, \{S_{w,x}\ |\ w\in V_1\cup V_2, x\in \Sigma\}) $
\end{minipage}}
\end{center}

The following theorem is shown in~\cite{reduction-lecture-notes}. We give a proof below for the sake of completeness.
\begin{theorem}
\label{lem:uw} If all edges of $G$ can be satisfied then the instance of \mbox{Set Cover} constructed has a set cover of
size $|V_1|+|V_2|$. On the other hand if at most $\frac{2}{\ell^2}$-fraction of edges of $G$ can be satisfied then the minimum
size of set cover for the \mbox{Set Cover} instance constructed above is $\frac{\ell}{8}(|V_1|+|V_2|)$.
\end{theorem}

Assuming Theorem~\ref{lem:uw}, we obtain Theorem~\ref{redn:proj-set-cover} by setting $\epsilon=\frac{2}{\ell^2}$ in
Theorem~\ref{lem:uw}. Recall that $m= |\Sigma|$. Hence the size of the universe is $|E\times B| = |E|\times |B| = |E|\times
2^{O(\frac{1}{\sqrt{\epsilon}})}\times |\Sigma|^2$ and the number of sets is $|\Sigma|\times (|V_1|+|V_2|)$.

We prove Theorem~\ref{lem:uw} via the following two lemmas:

\begin{lemma}
\label{lem:completeness} If all the edges of $G$ can be satisfied then the instance of Set Cover has a set cover of size
$|V_1| + |V_2|$
\end{lemma}
\begin{proof}
Let $\delta: V_1 \cup V_2 \rightarrow \Sigma$ denote a labeling for $G$ that satisfies all the edges $E$. Pick the following
set of sets $\mathcal{S} = {S_{w,\delta(w)}\ |\ w\in V_1 \cup V_2}$. The number of sets in $\mathcal{S}$ is $|V_1| + |V_2|$.
We claim that $\mathcal{S}$ is a valid set cover for $E\times B$. For every edge $e = (u, v0$ we show the following holds
\begin{equation}
\{e\}\times B \subseteq S_{u,\delta(u)}\cup S_{v,\delta(v)}
\label{eqn:1}
\end{equation}
The definition of $S_{u,\delta(u)}$ and $S_{v,\delta(v)}$ implies
\begin{equation}
\{e\}\times C_{\delta(v)} \subseteq S_{v,\delta(v)}\ \text{and}\ \{e\}\times \overline{C_{\Pi_{e}(\delta(u))}} \subseteq S_{u,\delta(u)}
\label{eqn:2}
\end{equation}
Since $\delta$ satisfies all the edges (and hence also satisfies $e$), we have $\Pi_{e}(\delta(u))=\delta(v)$. Therefore
\begin{equation}
\label{eqn:3}
\{e\}\times \overline{C_{\delta(v)}} = \{e\}\times \overline{C_{\Pi_{e}(\delta(u))}} \subseteq S_{u,\delta(u)}
\end{equation}
Now we can see that Equation~\ref{eqn:2} and Equation~\ref{eqn:3} imply Equation~\ref{eqn:1}. Moreover, taking the union of
the containment relation implied by Equation~\ref{eqn:1} for all edges $e$, we get $E\times B \subseteq \bigcup_{u\in V_1\cup
V_2} S_{u,\delta(u)}$ which completes the proof.
\end{proof}

\begin{lemma}
\label{lem:soundness} If at most $\frac{2}{\ell^2}$-fraction of edges of $G$ can be satisfied then the minimum size of set
cover for the \mbox{Set Cover} instance is $\frac{\ell}{8}(|V_1|+|V_2|)$.
\end{lemma}
\begin{proof}
We prove the contrapositive. Suppose there is a set cover $\mathcal{S}$ with $|\mathcal{S}| < \frac{\ell}{8}(|V_1| + |V_2|)$.
Then for each vertex $w$ define the set of labels
$$ L_{w}  = \{ c\in \Sigma\ |\ S_{w,c}\in \mathcal{S}\} $$
This implies that $|\mathcal{S}|=\sum_{w\in V_1\cup V_2} |L_w|$. Hence the average cardinality of $L_w$ satisfies
$$ \dfrac{\sum_{w\in V_1\cup V_2}}{|V_1|+|V_2|} = \dfrac{|\mathcal{S}|}{|V_1|+|V_2|} \leq \dfrac{\ell}{8}  $$
If there are more than $\frac{\ell}{4}$ vertices such that $|L_w|>\frac{\ell}{2}$, then the total sum $\sum_{w} |L_w|$ would
be greater than $\frac{\ell}{8}$, which is a contradiction. Hence at least $\frac{3}{4}$ of the vertices $w\in V_1\cup V_2$
satisfy $|L_w|\leq \frac{1}{2}$. Since we can assume that the bipartite graph of the Projection Games instance is regular, we
have that at least half the edges have both endpoints $(u,v)$ such that $|L_u|<\frac{\ell}{2}$ and $|L_v|<\frac{\ell}{2}$.
\begin{definition}
We say that an edge $e=(u,v)$ is \emph{frugally covered} if $|L_u|<\frac{\ell}{2}$ and $|L_v|<\frac{\ell}{2}$.
\end{definition}
Consider the following labeling $\delta'$ for $G$: for every $w\in V_1\cup V_2$ choose an element from $L_w$ uniformly at
random. We now show that the expected fraction of edges covered by $\delta'$ is at least $\frac{2}{\ell^2}$, which will
complete the proof.

To show this, we obtain that each \emph{frugally covered} edge is satisfied by $\delta'$ with probability at least
$\frac{4}{\ell^2}$. Since there are at least $\frac{|E|}{2}$ \emph{frugally covered} edges, we are done. It remains to show
that any frugally covered edge is satisfied by $\delta'$ with probability at least $\frac{4}{\ell^2}$. Let $e=(u,v)$ be any
frugally covered edge. Let $L_{u}=\{a_1, a_2, \ldots, a_p\}$ and $L_v = \{b_1, b_2, \ldots, b_q\}$. Since $e$ is frugally
covered we have $\frac{\ell}{2}>\max\{p,q\}$. The sets in $\mathcal{S}$ completely cover $E \times B$, and hence they also
cover $e\times  B$. Note, that for any vertex $w\notin \{u,v\}$ we have $|S_{w,x} \cap \{e × B\}| = 0$ for all $x\in \Sigma$.
In other words, no element of the set $e\times B$ can be covered by any of the sets $S_{w,x}$ for any vertex $w \notin \{u,
v\}$. Therefore the set $e\times B$ is covered by the sets chosen for vertices $u$ and $v$. That is,
$$ \{e\}\times B \subseteq (\bigcup_{i=1}^{p} S_{u,a_i}) \cup (\bigcup_{j=1}^{q} S_{v,b_j}) $$
By definition of $S_{v,x}$, we have $S_{v,x}\cap (\{e\}\times B) = \{e\}\times C_x$. Similarly $S_{u,y}\cap (\{e\}\times B) =
\{e\}\times \overline{C}_{\Pi_{e}(y)}$. Restricting the sets $S_{u, a_i}$ and $S_{v, b_j}$ to $\{e\}\times B$ in the above
containment we get
$$ \{e\}\times B \subseteq \Big( \{e\}\times \bigcup_{i=1}^{p} \overline{C}_{\Pi(a_i)} \Big) \cup \Big( \{e\}\times \bigcup_{j=1}^{q} C_{b_j} \Big) $$
Therefore, we have
$$  B\subseteq \Big( \bigcup_{i=1}^{p} \overline{C}_{\Pi(a_i)} \Big) \cup \Big( \bigcup_{j=1}^{q} C_{b_j} \Big) $$
This means that $B$ is covered by a family of $p+q\leq \ell$ sets, all of which are either of the form $C_i$ or
$\overline{C}_i$. Since $(B, C_i)$ form a $(m,\ell)$ set system, there exists an index $i$ such that both $C_i$ and
$\overline{C_i}$ are present among the $p+q$ sets. This implies for some $a_i, b_j$ we have $\Pi_{e}(a_i)=b_j$. Since we
choose the labels uniformly at random, with probability $\frac{1}{pq}$ we choose both $\delta'(u)=a_i$ and $\delta'(v)=b_j$.
Thus the probability that $e$ is satisfied by $\delta'$ is at least $\frac{1}{pq}\geq (\frac{2}{\ell})^2 = \frac{4}{\ell^2}$.
\end{proof}

\subsection{Composing the Two Reductions:}

Composing the reductions from Conjecture~\ref{conj:pgc} and Theorem~\ref{redn:proj-set-cover} we get:

\begin{theorem}
\label{thm:sat-set-cover} There exists $c>0$, such that for every $\epsilon>0$ there is a reduction \texttt{RED-3} from SAT to
\mbox{Set Cover} which maps an instance $I$ of SAT to an instance $I_2$ of \mbox{Set Cover} such that
\begin{enumerate}
\item If a YES instance $I$ of SAT satisfies $|I|^c\geq \frac{1}{\epsilon}$, then $I_2$ has a set cover of size $\beta$.
\item If a NO instance $I$ of SAT satisfies $|I|^c\geq \frac{1}{\epsilon}$, then $I_2$ does not have a set cover of size
    less than $\frac{\beta}{\sqrt{32\epsilon}}$.
\item The size $N$ of the universe for the instance $I'$ is $2^{O(\frac{1}{\sqrt{\epsilon}})} \times
    \emph{poly}(\frac{1}{\epsilon})\times \emph{poly}(|I|)$.
\item The number $M$ of sets for the set cover instance $I'$ is $\emph{poly}(\frac{1}{\epsilon})\times \emph{poly}(|I|)$.
\item The total time required for \texttt{RED-3} is $\text{emph}(|I|) + 2^{O(\frac{1}{\sqrt{\epsilon}})} \times
    \emph{poly}(\frac{1}{\epsilon})\times \emph{poly}(|I|)$.
\end{enumerate}
where $\beta\leq |I_1|=|I|^{1+o(1)}\cdot \emph{poly}(\frac{1}{\epsilon})$. Note that the number of elements is very large compared to the number of sets.
\end{theorem}
\begin{proof}
We apply the reduction from Theorem~\ref{redn:proj-set-cover} with $|\Sigma| = \text{poly}(\frac{1}{\epsilon})$ and
$|V_1|+|V_2|+|E|= n = |I|^{1+{o(1)}}\cdot \text{poly}(\frac{1}{\epsilon})$. Substituting these values in
Conjecture~\ref{conj:pgc} and Theorem~\ref{redn:proj-set-cover}, we get the parameters as described in the given theorem. We
work out each of the values below:
\begin{enumerate}
\item If $I$ is a YES instance of SAT satisfying $\epsilon\geq \frac{1}{|I|^c}$, then \texttt{RED-1} maps it to an
    instance $I_1=(G=(V_1,V_2,E), \Sigma, \pi)$ of Projection Games such that all edges of $I_1$ can be satisfied. Then
    \texttt{RED-2} maps $I_1$ to an instance $I_2$ of Set Cover such that $I_2$ has a set cover of size $\beta=|V_1|+|V_2|
    \leq |V_1|+|V_2|+|E| = |I_1| = |I|^{1+o(1)}\cdot \text{poly}(\frac{1}{\epsilon})$.
\item If $I$ is a NO instance of SAT satisfying $\epsilon\geq \frac{1}{|I|^c}$, then \texttt{RED-1} maps it to an instance
    $I_1=(G=(V_1,V_2,E), \Sigma, \pi)$ of Projection Games such that at most $\epsilon$-fraction of the edges of $I_1$ can
    be satisfied. Then \texttt{RED-2} maps $I_1$ to an instance $I_2$ of Set Cover such that $I_2$ does not have a set
    cover of size $\frac{\beta}{\sqrt{32 \epsilon}}$, where $\beta$ is as calculated above.
\item By Theorem~\ref{redn:proj-set-cover}(3), the size of the universe is $2^{O(\frac{1}{\sqrt{\epsilon}})} \times
    |\Sigma|^{2}\times |E|$. Observing that $|\Sigma|= \text{poly}(\frac{1}{\epsilon})$ and $|E|\leq |I_1|=
    |I|^{1+{o(1)}}\cdot \text{poly}(\frac{1}{\epsilon})$, it follows that the size of the universe is
    $2^{O(\frac{1}{\sqrt{\epsilon}})} \times \text{poly}(\frac{1}{\epsilon})\times \text{poly}(|I|)$.
\item By Theorem~\ref{redn:proj-set-cover}(3), the number of sets is $|\Sigma|\times (|V_1|+|V_2|)$. Observing that
    $|\Sigma| = \text{poly}(\frac{1}{\epsilon})$ and $|V_1|+|V_2|\leq |I_1|= |I|^{1+{o(1)}}\cdot
    \text{poly}(\frac{1}{\epsilon})$, it follows that the number of sets is $\text{poly}(\frac{1}{\epsilon})\times
    \text{poly}(|I|)$.
\item Since \texttt{RED-3} is the composition of \texttt{RED-1} and \texttt{RED-2}, the time required for \texttt{RED-3}
    is the summation of the times required for \texttt{RED-1} and \texttt{RED-2}. By Conjecture~\ref{conj:pgc}, the time
    required for \texttt{RED-1} is $\text{poly}(|I|)$. By Theorem~\ref{redn:proj-set-cover}(4), the time required for
    \texttt{RED-2} is at most $2^{O(\frac{1}{\sqrt{\epsilon}})} \times \text{poly}(|\Sigma|)\times
    \text{poly}(|E|+|V_1|+|V_2|)$. Observing that $|\Sigma|=\text{poly}(\frac{1}{\epsilon})$ and $|V_1|+|V_2|+|E| = |I_1|
    = |I|^{1+o(1)}\cdot \text{poly}(\frac{1}{\epsilon})$, it follows that the time required for \texttt{RED-2} is at most
    $2^{O(\frac{1}{\sqrt{\epsilon}})} \times \text{poly}(\frac{1}{\epsilon})\times \text{poly}(|I|)$. Adding up the two,
    the time required for \texttt{RED-3} is at most $\text{poly}(|I|) + 2^{O(\frac{1}{\sqrt{\epsilon}})} \times
    \text{poly}(\frac{1}{\epsilon})\times \text{poly}(|I|)$.
\end{enumerate}
\end{proof}

Finally we are ready to prove Theorem~\ref{setcover}.

\subsection{Proof of Theorem~\ref{setcover}(1)}

The roadmap of the proof is as follows: suppose to the contrary there exists an FPT \textbf{optimum} approximation algorithm,
say $\mathbb{A}$, for Set Cover with ratio $\rho(OPT)=OPT^{F_1}$ in $2^{OPT^{F_2}}\cdot \text{poly}(N,M)$ time, where $N$ is
the size of the universe and $M$ is the number of sets (recall Definition~\ref{defn-opt}). We will choose the constant $F_1$
such that using \texttt{RED-3} from Theorem~\ref{thm:sat-set-cover} (which assumes PGC), the algorithm $\mathbb{A}$ applied to
the instance $I_2$ will be able to decide the instance $I_1$ of SAT. Then to violate ETH we will choose the constant $F_2$
such that the running time of $\mathbb{A}$ summed up with the time required for \texttt{RED-3} is subexponential in $|I|$.

Let $c>0$ be the constant from Conjecture~\ref{conj:pgc}. Fix some constant $1>\delta>0$ and let $c^{*}=\min\{c,2-2\delta\}$.
Note that $\frac{c^{*}}{2}\leq 1-\delta$. Choosing $\epsilon = \frac{1}{|I|^{c^{*}}}$ implies $\epsilon\geq \frac{1}{|I|^c}$,
since $c\geq c^*$. We carry out the reduction \texttt{RED-3} given by Theorem~\ref{thm:sat-set-cover}. From
Conjecture~\ref{conj:pgc}(3), we know that $|I_1| = |I|^{1+o(1)}\cdot \text{poly}(\frac{1}{\epsilon})$. Let $\lambda>0$ be a
constant such that the $\text{poly}(\frac{1}{\epsilon})$ is upper bounded by $(\frac{1}{\epsilon})^{\lambda}$. Then
Theorem~\ref{thm:sat-set-cover} implies $\beta \leq |I|^{1+o(1)}\cdot (\frac{1}{\epsilon})^{\lambda}$. However we have chosen
$\epsilon = \frac{1}{|I|^{c^*}}$, and hence asymptotically we get
\begin{equation}
\beta \leq |I|^{2}\cdot |I|^{\lambda c^{*}} = |I|^{2+\lambda c^{*}}
\label{eqn:half}
\end{equation}
Choose the constant $F_1$ such that $\dfrac{c^*}{4(2+\lambda c^*)}\geq F_1$. Suppose Set Cover has an FPT optimum
approximation algorithm $\mathbb{A}$ with ratio $\rho(OPT)=OPT^{F_1}$ (recall Definition~\ref{defn-opt}). We show that this
algorithm $\mathbb{A}$ can decide the SAT problem. Consider an instance $I$ of SAT, and let $I_2=$\texttt{RED-3}$(I)$ be the
corresponding instance of Set Cover. Run the FPT approximation algorithm on $I_G$, and let $\mathbb{A}(I_2)$ denote the output
of $\textsc{ALG}$. We have the following two cases:
\begin{itemize}
\item $\frac{\beta}{\sqrt{32 \epsilon}}\leq \mathbb{A}(I_2)$: Then we claim that $I$ is a NO instance of SAT. Suppose to
    the contrary that $I$ is a YES instance of SAT. Then Theorem~\ref{thm:sat-set-cover}(1) implies $\beta \geq OPT(I_2)$.
    Hence $\frac{\beta}{\sqrt{32 \epsilon}}\leq \mathbb{A}(I_2)\leq OPT\cdot \rho(OPT) = OPT^{1+F_{1}} = \beta^{1+F_{1}}
    \Rightarrow \frac{1}{\sqrt{32 \epsilon}}\leq \beta^{F_1}$. However, asymptotically we have $\frac{1}{\sqrt{32
    \epsilon}} = \frac{|I|^{\frac{c^*}{2}}}{\sqrt{32}} > |I|^{\frac{c^*}{4}} \geq (|I|^{2+\lambda c^*})^{F_1} =
    \beta^{F_1}$, where the last two inequalities follows from the choice of $F_1$ and Equation~\ref{eqn:half}
    respectively. This leads to a contradiction, and therefore $I$ is a NO instance of SAT.
\item $\frac{\beta}{\sqrt{32 \epsilon}}> \mathbb{A}(I_2)$: Then we claim that $I$ is a YES instance of SAT. Suppose to the
    contrary that $I$ is a NO instance of SAT. Then Theorem~\ref{thm:sat-set-cover}(2) implies $OPT(I_2)\geq
    \frac{\beta}{\sqrt{32 \epsilon}}$. Therefore we have $\frac{\beta}{\sqrt{32 \epsilon}}> \mathbb{A}(I_2) \geq
    OPT(I_2)\geq \frac{\beta}{\sqrt{32 \epsilon}}$.
\end{itemize}
Therefore we run the algorithm $\mathbb{A}$ on the instance $I_2$ and compare the output $\frac{\beta}{\sqrt{32 \epsilon}}$
with $n^{\epsilon}$. As seen above, this comparison allows us to decide the SAT problem.

We now choose the constant $F_2$ such that the running time of $\mathbb{A}$ summed up with the time required for
\texttt{RED-3} is subexponential in $|I|$.

By Theorem~\ref{thm:sat-set-cover}(5), the total time taken by \texttt{RED-3} is $\text{poly}(|I|) +
2^{O(\frac{1}{\sqrt{\epsilon}})} \times \text{poly}(\frac{1}{\epsilon})\times \text{poly}(|I|) = \text{poly}(|I|) +
2^{O(|I|^{\frac{c^*}{2}})} \times \text{poly}(|I|^{\frac{c^{*}}{2}})\times \text{poly}(|I|) = \text{poly}(|I|) + 2^{o(I)}\cdot
\text{poly}(|I|)$ since $\frac{c^{*}}{2}\leq 1-\delta$. Hence total time taken by \texttt{RED-3} is subexponential in $I$. We
now show that there exists a constant $F_2$ such that the claimed running time of $2^{OPT^{F_2}}\cdot \text{poly}(N,M)$ for
the algorithm $\mathbb{A}$ is subexponential in $|I|$, thus contradicting ETH. We do not have to worry about the
$\text{poly}(N,M)$ factor: the reduction time is subexponential in $|I|$, and hence $\max\{N,M\}$ is also upper bounded by a
subexponential function of $|I|$. Hence, we essentially want to choose a constant $F_2>0$ such that $OPT^{F_2}\leq
M^{F_2}=o(|I|)$. From Theorem~\ref{thm:sat-set-cover}(4), we know that $M\leq |\Sigma|\times |V_1 + V_2|$. Since
$|\Sigma|=\text{poly}(\frac{1}{\epsilon})$, let $\alpha>0$ be a constant such that the $|\Sigma|\leq
(\frac{1}{\epsilon})^{\alpha}$. We have seen earlier in the proof that $|V_1 + V_2|\leq |I_1|\leq |I|^2 \cdot
(\frac{1}{\epsilon})^{\lambda} = |I|^{2+c^* \lambda}$. Therefore $M^{F_2}\leq (|I|^{2+c^* \lambda + c^* \alpha})^{F_2}$.
Choosing $F_2< \frac{1}{2+\lambda c^* + 2\alpha c^*}$ gives $\OPT^{F_2} = o(|I|)$, which is what we wanted to show.

\qed

\subsection{Proof of Theorem~\ref{setcover}(2)}

Observe that due to Theorem~\ref{thm:mini-implication}, Theorem~\ref{setcover}(1) implies Theorem~\ref{setcover}(2).

\section{An FPT Inapproximability Result for \mbox{Clique}}
\label{sec:lower-bound-clique}

We use the following theorem due to Zuckerman~\cite{Zu}, which in turn is a derandomization of a result of H{\.a}stad~\cite{hast} .

\begin{theorem}{\em~\cite{hast,Zu}}
Let $X$ be any problem in $\NP$. For any constant $\epsilon>0$ there exists a polynomial time reduction from $X$ to
\mbox{Clique} so that the gap between the clique sizes corresponding to the YES and NO instances of $X$ is at least
$n^{1-\epsilon}$, where $n$ is the number of vertices of the Clique instances.
\end{theorem}

\subsection{Proof of Theorem~\ref{clique}(1)}

Fix a constant $1>\delta>0$. Set $0<\epsilon = \frac{\delta}{\delta+2}$, or equivalently $\delta =
\frac{2\epsilon}{1-\epsilon}$. Let $X$ be any problem in $\NP$. Let the Hastad-Zuckerman reduction from $X$ to Clique
\cite{hast,Zu} which creates a gap of at least $n^{1-\epsilon}$ map an instance $I$ of $X$ to the corresponding instance $I_G$
of Clique. Since the reduction is polynomial, we know that $n=|I_G|=|I|^D$ for some constant $D(\epsilon)>0$. Note that $D$
depends on $\epsilon$, which in turn depends on $\delta$. Hence, ultimately $D$ depends on $\delta$. If $I$ is a YES instance
of $X$, then $I_G$ contains a clique of size at least $n^{1-\epsilon}$ since each graph has a trivial clique of size one and
the gap between YES and NO instances of Clique is at least $n^{1-\epsilon}$. Similarly, observe that a graph on $n$ vertices
can have a clique of size at most $n$. To maintain the gap of at least $n^{1-\epsilon}$, it follows if $I$ is a NO instance of
$X$ then the maximum size of a clique in $I_G$ is at most $n^{\epsilon}$. To summarize, we have
\begin{itemize}
\item If $I$ is a YES instance, then $OPT(I_G)\geq n^{1-\epsilon}$
\item If $I$ is a NO instance, then $OPT(I_G)\leq n^{\epsilon}$
\end{itemize}

Suppose Clique has an FPT \textbf{optimum} approximation algorithm $\mathbb{A}$ with ratio $\rho(OPT)=OPT^{1-\delta}$ (recall
Definition~\ref{defn-opt}). We show that this algorithm $\mathbb{A}$ can decide the problem $X$. Consider an instance $I$ of
$X$, and let $I_G$ be the corresponding instance of Clique. Run the FPT approximation algorithm on $I_G$, and let
$\mathbb{A}(I_G)$ denote the output of $\mathbb{A}$. We have the following two cases:
\begin{itemize}
\item \underline{$n^{\epsilon}\geq \mathbb{A}(I_G)$}: Then we claim that $I$ is a NO instance of $X$. Suppose to the
    contrary that $I$ is a YES instance of $X$, then we have $n^{\epsilon}\geq \mathbb{A}(I_G)\geq
    \frac{OPT_{I_G}}{\rho(OPT(I_G)))} = (OPT(I_G))^{\delta} \geq (n^{1-\epsilon})^{\delta} = n^{2\epsilon}$, which is a
    contradiction.
\item \underline{$n^{\epsilon}<\mathbb{A}(I_G)$}: Then we claim that $I$ is a YES instance of $X$. Suppose to the contrary
    that $I$ is a NO instance of $X$, then we have $n^{\epsilon}<\mathbb{A}(I_G)\leq OPT(I_G)\leq n^{\epsilon}$, which is
    a contradiction.
\end{itemize}
We run the algorithm $\mathbb{A}$ on the instance $I_G$ and compare the output $\mathbb{A}(I_G)$ with
$n^{\epsilon}$. As seen above, this comparison allows us to decide the problem $X$. We now show how to choose the constant $F$
such that the running $2^{OPT^F}\cdot \text{poly}(n)$ is subexponential in $|I|$. We claim that $F = \frac{1}{D + 1}$ works.
Note that $OPT(I_G)\leq n$ always. Hence $2^{OPT^F}\cdot \text{poly}(n) \leq 2^{n^{F}}\cdot \text{poly}(n) =
2^{(|I|^{D})^{F}}\cdot \text{poly}(|I|^D) = 2^{|I|^{DF}}\cdot \text{poly}(|I|) = 2^{|I|^{\frac{D}{D+1}}}\cdot
\text{poly}(|I|)= 2^{o(I)}\cdot \text{poly}(|I|)$.
This implies we can could solve $X$ in subexponential time using $\mathbb{A}$. However $X$ was any problem chosen from the
class $\NP$, and hence $\NP\subseteq \SUBEXP$. \qed

\subsection{Proof of Theorem~\ref{clique}(2)}

Fix a constant $1>\delta>0$. Set $0<\epsilon = \frac{\delta}{\delta+1}$, or equivalently $\delta =
\frac{\epsilon}{1-\epsilon}$. Let $X$ be any problem in $\NP$. Let the Hastad-Zuckerman reduction from $X$ to Clique
\cite{hast,Zu} which creates a gap of at least $n^{1-\epsilon}$ map an instance $I$ of $X$ to the corresponding instance $I_G$
of Clique. Since the reduction is polynomial, we know that $n=|I_G|=|I|^D$ for some constant $D(\epsilon)>0$. Note that $D$
depends on $\epsilon$, which in turn depends on $\delta$. Hence, ultimately $D$ depends on $\delta$. If $I$ is a YES instance
of $X$, then $I_G$ contains a clique of size at least $n^{1-\epsilon}$ since each graph has a trivial clique of size one and
the gap between YES and NO instances of Clique is at least $n^{1-\epsilon}$. Similarly, observe that a graph on $n$ vertices
can have a clique of size at most $n$. To maintain the gap of at least $n^{1-\epsilon}$, it follows if $I$ is a NO instance of
$X$ then the maximum size of a clique in $I_G$ is at most $n^{\epsilon}$.

Suppose Clique has an FPT approximation algorithm $\textsc{ALG}$ with ratio $\rho(k)=k^{1-\delta}$ (recall
Definition~\ref{defn-1}). We show that this algorithm $\textsc{ALG}$ can decide the problem $X$. Set $k=n^{\epsilon}$. On the
input $(I_G, n^{\epsilon})$ to \textsc{ALG}, there are two possible outputs:
\begin{itemize}
\item \textsc{ALG} outputs \texttt{reject} $\Rightarrow OPT(I_G)< n^{\epsilon} \Rightarrow I$ is a NO instance of $X$
\item \textsc{ALG} outputs a clique of size $\geq \frac{k}{\rho(k)}\Rightarrow OPT(I_G)\geq \frac{k}{\rho(k)} =
    \frac{k}{k^{1-\delta}} = k^{\delta} = (n^{\epsilon})^{\delta}  = n^{1-\epsilon} \\ \Rightarrow I$ is a YES instance of
    $X$
\end{itemize}

Therefore the FPT approximation algorithm $\textsc{ALG}$ can decide the problem $X\in \NP$.

We now show how to choose the constant $F'$ such that the running $\text{exp}(k^{F'})\cdot \text{poly}(n)$ is subexponential
in $|I|$. We claim that $F' = \frac{1}{\epsilon\cdot D + 1}$ works. This is because $2^{k^{F'}}\cdot \text{poly}(n) =
2^{n^{\epsilon F'}}\cdot \text{poly}(n) = 2^{|I|^{\epsilon DF'}}\cdot \text{poly}(|I|^D) = 2^{|I|^{\frac{\epsilon D}{\epsilon
D + 1}}}\cdot \text{poly}(|I|) = 2^{o(I)}\cdot \text{poly}(|I|)$.

This implies we can could solve $X$ in subexponential time using \textsc{ALG}. However $X$ was any problem chosen from the
class $\NP$, and hence $\NP\subseteq \SUBEXP$. \qed

\section{Polytime $f(OPT)$-approximation for W[1]-hard problems}
\label{sec:conceptual-differences}

In Section~\ref{sec:intro-f-opt-poly-time} we have seen the motivation for designing polynomial time $f(OPT)$-approximation
algorithms for W[1]-hard problems such as \mec, \scss, \dsf and Directed Steiner Network. Our results are summarized in Figure~\ref{f(opt)-table}.

\begin{figure}
\begin{center}
    \begin{tabular}{ | l | l | l | l |}
    \hline
     & W[1]-hardness & Polytime Approx. Ratio \\ \hline
    Strongly Connected Steiner Forest & Guo et al.~\cite{guo-directed-steiner} & $OPT^{\epsilon}$ (Lemma~\ref{lem:scss-approx-poly})\\ \hline
    Directed Steiner Forest & Lemma~\ref{lem:dsf-hardness} & $OPT^{1+\epsilon}$ (Lemma~\ref{lem:dsf-approx})\\ \hline
    Directed Steiner Network & Lemma~\ref{lem:dsn-hardness} & $OPT^2$ (Lemma~\ref{lem:dsn-approx}) \\ \hline
    Minimum Edge Cover & Lemma~\ref{lem:mec-hardness} & $OPT-1$ (Lemma~\ref{lem:mec-approx}) \\ \hline
    Directed Multicut & Marx and Razgon~\cite{m3} & $3\cdot OPT$~\cite{gup}\\ \hline
    \end{tabular}
\end{center}
\caption{Polytime $f(OPT)$-approximation for W[1]-hard problems} \label{f(opt)-table}
\end{figure}

\subsection{The \scss Problem}

\begin{lemma}
For any constant $\epsilon > 0$, the \scss problem has a $2\cdot OPT^{\epsilon}$-approximation in polynomial time.
\label{lem:scss-approx-poly}
\end{lemma}
\begin{proof}
Fix any constant $\epsilon>0$. Let $G_{rev}$ denote the reverse graph obtained from $G$, i.e., reverse the orientation of each
edge. Any solution of the Strongly Connected Steiner Subgraph instance must contain a path from $t_1$ to each terminal in
$T\setminus t_1$ and vice versa. Consider the following two instances of the \dst problem: $I_1=(G,t_1,T\setminus t_1)$ and
$I_2=(G_{rev},t_1,T\setminus t_1)$. In~\cite{CCC} an $|T|^{\epsilon}$-approximation is designed for \dst in polynomial time,
for any constant $\epsilon>0$. Let $E_1,E_2$ be the $|T|^{\epsilon}$-approximate solutions for the two instances and say that
their optimum solutions are $OPT_1, OPT_2$ respectively. Let OPT be the size of optimum solution for the \scss instance, then
clearly $|OPT| \geq \max\{|OPT_1|,|OPT_2|\}$. Clearly $E_1\cup E_2$ is a solution for the \scss instance as $E_j$ is a
solution for $I_j$ for $1\leq j\leq 2$. It now remains to bound the size of this solution: $|E_1\cup E_2|\leq |E_1|+|E_2| \leq
|T|^{\epsilon}|OPT_1| + |T|^{\epsilon}|OPT_2| = |T|^{\epsilon}(|OPT_1| + |OPT_2|) \leq 2|T|^{\epsilon}|OPT|$. As every
terminal has at least one incoming edge (and these edges are pairwise disjoint) we get that $OPT\geq |T|=k$. Therefore
$|T|^\epsilon\leq OPT^\epsilon$ which implies a $2\cdot OPT^\epsilon$-approximation factor.
\end{proof}

\subsection{The \dsf Problem} \label{sec:dsf}

The \dsf problem is LabelCover hard and thus admits no $2^{\log^{1-\epsilon} n}$-approximation for any constant $\epsilon$
\cite{dodis}. The best know approximation factor for the problem is $n^{\frac{2}{3}}$ \cite{FGZ,pib}. We now define the
problem formally:

\begin{center}
\noindent\framebox{\begin{minipage}{6.00in}
\mbox{\textbf{\dsf}}\\
\emph{Input }: A digraph $G=(V,E)$ and a set of terminals $T=\{(s_1,t_1),\ldots,(s_k,t_k)\}$.\\
\emph{Problem}: Does there exist a set $E'\subseteq E$ such that $|E'|\leq p$ and $(V,E')$ has a $s_i\rightarrow t_i$ path for
every $i\in
[k]$.\\
\emph{Parameter}: $p$
\end{minipage}}
\end{center}

\begin{lemma}
The \dsf problem is W[1]-hard parameterized by solution size plus number of terminal pairs.
\label{lem:dsf-hardness}
\end{lemma}
\begin{proof}
We give a reduction from the \scss problem. Consider an instance $(G,T,p)$ of \scss where $T=\{t_1,t_2,\ldots,t_{\ell}\}$. We
now build a new graph $G^{*}$ as follows:
\begin{itemize}
\item Add $2\ell$ new vertices: for every $i\in \ell$, we introduce vertices $r_i$ and $s_i$.
\item For every $i\in [\ell]$, add the edges $(r_i,t_i)$ and $(t_i,s_i)$.
\end{itemize}

Let the terminal pairs be $T^{*}=\{(r_i,s_j)\ |\ 1\leq i,j\leq \ell ; i\neq j \}$. We claim that the Strongly Connected
Steiner Subgraph instance $(G,T)$ has a solution of size $p$ if and only if there is a solution for the \dsf instance
$(G^{*},T^{*})$ of size $p+2\ell$.

Suppose there is a solution for the \scss instance of size $p$. Adding the edges from $E(G^{*})\setminus E(G)$ clearly gives a
solution for the \dsf instance of size $p+2\ell$. Conversely, suppose we have a solution for the \dsf instance of size
$p+2\ell$. Since $t_i$ is the only out-neighbor, in-neighbor of $r_i, s_i$ respectively the solution must contain all the
edges from $E(G^{*})\setminus E(G)$. Removing these edges clearly gives a solution of size $p$ to the \scss instance. Note
that $|T^{*}|=\ell(\ell-1)$. Since \scss is W[1]-hard parameterized by solution size plus number of terminals, we have that
\dsf is W[1]-hard parameterized by solution size plus number of terminal pairs.
\end{proof}

\begin{lemma}
The \dsf problem admits an $OPT^{1+\epsilon}$-approximation in polynomial time. \label{lem:dsf-approx}
\end{lemma}
\begin{proof}
Let $S=\{v\ |\ \exists\ x\ \text{such that}\ (v,x)\in T\}$. For every $v\in S$, let $T_v = \{x\ |\ (v,x)\in T\}$. For each $v\in S$, let the optimum for the instance $(G,v, T_v)$ of \dst be $OPT_v$. Clearly $OPT_v\leq OPT$, where $OPT$ is the optimum of the given \dsf instance.For each $v\in S$, we take the $|T_v|^{\epsilon}$-approximation given in~\cite{CCC} for the instance $(G, v, T_v)$ of \dst, and output the union of all these Steiner trees. Clearly this gives a feasible solution. We now analyze the cost.

Since each vertex in $S$ must have its own outgoing edge in the solution, we have $|S|\leq OPT$. Similarly, for every $v\in S$ each vertex of $T_v$ must have its own incoming edge in the solution, and hence $|T_v|\leq OPT$. Hence the cost of our solution is upper bounded by $\sum_{v\in S} |T_v|^{\epsilon}\cdot OPT_v \leq \sum_{v\in S} OPT^{\epsilon}\cdot OPT \leq |S|\cdot OPT^{1+\epsilon} \leq OPT^{2+\epsilon}$. Therefore, we get a $OPT^{1+\epsilon}$-approximation.
\end{proof}

\subsection{The \dsn Problem} \label{sec:dsn}

The \dsn problem is not known to admit any non-trivial approximation and of course is LabelCover hard. We define the problem
formally:

\begin{center}
\noindent\framebox{\begin{minipage}{6.00in}
\mbox{\textbf{\dsn}}\\
\emph{Input }: A digraph $G=(V,E)$, a set of terminals $T=\{(s_1,t_1),\ldots,(s_k,t_k)\}$, a demand $d_i$ between $s_i,t_i$
for every $i\in [k]$\\
\emph{Problem}: Does there exist a set $E'\subseteq E$ such that $|E'|\leq p$ and $(V,E')$ has $d_i$ disjoint $s_i\rightarrow
t_i$ paths for
every $i\in [k]$. \\
\emph{Parameter}: $p$
\end{minipage}}
\end{center}

\begin{lemma}
The \dsn problem is W[1]-hard parameterized by solution size plus number of terminal pairs.
\label{lem:dsn-hardness}
\end{lemma}
\begin{proof}
The lemma follows from Lemma~\ref{lem:dsf-hardness} and the fact that the \dsf problem is a special case of the \dsn problem
with $d_i=1$ for every $i$.
\end{proof}

\begin{lemma}
The \dsn problem admits an $OPT^2$-approximation in polynomial time. \label{lem:dsn-approx}
\end{lemma}
\begin{proof}
Let $S=\{v\ |\ \exists\ x\ \text{such that}\ (v,x)\in T\}$ and $S'=\{x\ |\ \exists\ v\ \text{such that}\ (v,x)\in T\}$. For every $v\in S$, let $T_v = \{x\ |\ (v,x)\in T\}$. For each $v\in S$ and $x\in T_V$, let the demand for the pair $(v,x)$ be $d_{vx}$. Make $v$ as a sink and $x$ as a source. Using min-cost max-flow, find the smallest edge set, say $E_{vx}$, such that there are $d_{vx}$ disjoint $v\rightarrow x$ paths in $G$. Clearly $|E_{vx}|\leq OPT$ since any solution for the \dsn instance must contain $d_{vx}$ disjoint $v\rightarrow x$ paths. Also, as seen before $OPT\geq |S|$ since each vertex in $S$ must have at least own outgoing edge in any solution. Similarly $OPT\geq |S'|$. We output $\bigcup_{v\in S, x\in T_v} E_{vx}$ as our solution. Clearly this is a feasible solution. Its cost is $\sum_{v\in S, x\in T_v} E_{vx} \leq \sum_{v\in S, x\in T_v} OPT \leq |S|\cdot |S'|\cdot OPT \leq OPT\cdot OPT\cdot OPT$, and hence this gives an $OPT^2$-approximation.
\end{proof}

\subsection{The \mec Problem} \label{sec:mec}

In this section, we show that the \mec problem is W[1]-hard parameterized by size of the solution, and it admits an
$(OPT-1)$-approximation in polynomial time. The best approximation for \mec is $O(n^{0.172})$ due to Chlamtac et
al.~\cite{chlamtac-mec}.

\begin{center}
\noindent\framebox{\begin{minipage}{6.00in}
\textbf{\mec}\\
\emph{Input }: A graph $G=(V,E)$ and an integer $k$\\
\emph{Problem} : Does there exists a set $S\subseteq V$ such that $|S|\leq p$ and the number of edges with both endpoints in
$S$ is at least $k$.\\
\emph{Parameter }: $p$
\end{minipage}}
\end{center}

To show that W[1]-hardness of \mec we reduce from the \mcc problem which is known to be W[1]-hard~\cite{multi-colored-clique}\footnote{Cai~\cite{DBLP:journals/cj/Cai08} has independently shown the W[1]-hardness of \mec with parameter $p$. They call this problem as \textsc{Maximum $p$-Vertex Subgraph}}:

\begin{lemma}
The \mec problem is W[1]-hard parameterized by the size of the solution. \label{lem:mec-hardness}
\end{lemma}
\begin{proof}
Given an instance $I_1=(G,\phi,p)$ of \mcc, we can consider another instance $I_2=(G,k,p)$ of \mec where $k=\binom{p}{2}$.
Clearly if $I_1$ is a YES instance, then $G$ has a multicolored clique and $I_2$ is a YES instance. In the other direction, if
$I_2$ is a YES instance then the $p$-sized set must form a clique in $G$, and must be in different color classes as $\phi$ is
a proper vertex coloring. This shows that \mec is W[1]-hard parameterized by the size of the covering set.
\qed
\end{proof}

Now we show give an approximation algorithm for the \mec problem.
\begin{lemma}
The \mec problem admits an $(OPT-1)$-approximation in polynomial time. \label{lem:mec-approx}
\end{lemma}
\begin{proof} Let $k$ be the desired number of edges in the solution and let $OPT$
be the minimum number of vertices required. If there is a feasible solution, then there must be at least $k$ edges in the graph. Pick any $k$ edges, and let $p'$ the size of the set which is the union of their endpoints. Clearly $p'\leq 2k$. Since $k\leq \frac{OPT(OPT-1)}{2}$, we have $p'\leq 2k\leq OPT(OPT-1)$, and hence we get a $(OPT-1)$-approximation\footnote{In fact, any feasible solution gives a $(OPT-1)$-approximation}.
\end{proof}

\section{Constant Factor FPT Approximation For SCSS}
\label{sec:scss}

In this section we show that SCSS has an FPT
2-approximation. We define the problem formally:

\begin{center}
\noindent\framebox{\begin{minipage}{6.00in}
\mbox{\textbf{\scss} (SCSS)}\\
\emph{Input }: An directed graph $G=(V,E)$, a set of terminals $T=\{t_1,t_2,\ldots,t_{\ell}\}$ and an integer $p$\\
\emph{Problem}: Does there exists a set $E'\subseteq E$ such that $|E'|\leq p$ and the graph $G'=(V,E')$ has a $t_{i}\rightarrow t_j$ path for every $i\neq j$\\
\emph{Parameter}: $p$
\end{minipage}}
\end{center}

\begin{lemma}
\scss has an FPT 2-approximation. \label{lem:scss-approx-fpt}
\end{lemma}
\begin{proof}
Let $G_{rev}$ denote the reverse graph obtained from $G$, i.e., reverse the orientation of each edge. Any solution of SCSS
instance must contain a path from $t_1$ to each terminal in $T\setminus t_1$ and vice versa. Consider the following two
instances of the \dst problem: $I_1=(G,t_1,T\setminus t_1)$ and $I_2=(G_{rev},t_1,T\setminus t_1)$, and let their optimum be
be $OPT_1, OPT_2$ respectively. Let $OPT$ be the optimum of given SCSS instance and $k$ be the parameter. If $OPT>k$ then we
output anything (see Definition~\ref{defn-1}). Otherwise we have $k\geq OPT\geq \max\{OPT_1,OPT_2\}$. We know that the \dst
problem is FPT parameterized by the size of the solution~\cite{dst-dreyfus}.
Hence we find the values $OPT_1, OPT_2$ in time which is FPT in $k$. Clearly the union of solutions for $I_1$ and $I_2$ os a
solution for instance $I$ of SCSS. The final observation is $OPT_1 + OPT_2 \leq OPT + OPT = 2\cdot OPT$.
\end{proof}

Guo et al.~\cite{guo-directed-steiner} show that SCSS is W[1]-hard parameterized by solution size plus number of terminals. It
is known that SCSS has no $\log^{2-\epsilon} n$-approximation in polynomial time for any fixed $\epsilon>0$, unless NP has quasi-polynomial Las Vegas algorithms~\cite{eran1}. Combining these facts with
Lemma~\ref{lem:scss-approx-fpt} implies that SCSS is a W[1]-hard problem that is not known to admit a constant factor
approximation in polynomial time but has a constant factor FPT approximation. This answers a question by
Marx~\cite{daniel-survey}. Previously the only such problem known was a variant of the \mbox{Almost-2-SAT}
problem~\cite{almost2sat-1} called \mbox{2-ASAT-BFL}, due to Marx and Razgon~\cite{daniel-esa-2-approx}.

\section{Open Problems}

In this paper, we have made some progress towards proving Conjecture~\ref{conj:clique-and-set-cover-no-fpt-approx}. We list two of the open problems below:
\begin{itemize}
\item Is there a W[2]-hard problem that admits an $f(OPT)$-approximation in polynomial time, for some function increasing $f$? In Section~\ref{sec:conceptual-differences}, we showed that various W[1]-hard problems admit $f(OPT)$-approximation algorithms in polynomial time, but no such W[2]-hard problem is known.
\item Is there a W[2]-hard problem that admits an FPT approximation algorithm with ratio $\rho$, for any function $\rho$? Grohe and Gr{\"u}ber~\cite{positive-2} showed that the W[1]-hard problem of finding $k$ vertex disjoint cycles in a directed graph has a FPT approximation with ratio $\rho$, for some computable function $\rho$. However, no such W[2]-hard problem is known.
\end{itemize}
It is known~\cite{positive-2,daniel-survey} that the existence of an FPT approximation algorithm with ratio $\rho$ implies that there is an $\rho'(OPT)$-approximation in polynomial time, for some function $\rho'$. Therefore, a positive answer to the first question implies a positive answer to the second question.

\bibliographystyle{abbrv}

\bibliography{u1}

\end{document}